\documentclass[10pt,leqno]{amsart}
\usepackage{titlesec}
\titleformat{\subsection}
{\normalfont\bfseries}
{\thesubsection}
{1em}
{}
\titleformat{\section}
{\normalfont\Large\bfseries}
{\thesection\quad}
{0pt}
{}

\usepackage{indentfirst,csquotes}
\usepackage{listings}
\usepackage{mathtools}
\usepackage{adjustbox}
\usepackage{stmaryrd}
\usepackage{graphicx}
\usepackage{ulem}
\usepackage[most]{tcolorbox}
\usepackage{mdframed}
\usepackage{amssymb}
\usepackage{algorithm}
\usepackage{algpseudocode}
\usepackage{amsthm}
\usepackage{amsmath}
\usepackage{pbox}
\usepackage{xcolor}
\usepackage{hyperref}
\usepackage{etoolbox}
\usepackage{cleveref}
\usepackage{tikz}
\usepackage{pgfplots}
\usepackage{amsfonts}
\usepackage{enumitem}
\usepackage{subcaption}

\input{preamble.sty}

\newtheorem{theorem}{Theorem}[section]

\newtheorem{problem}[theorem]{Problem}

\newtheorem{proposition}[theorem]{Proposition}

\theoremstyle{definition}
\newtheorem{example}[theorem]{Example}
\theoremstyle{plain}
\crefname{assumption}{Assumption}{Assumptions}
\crefname{problem}{Problem}{Problems}
\numberwithin{equation}{section}

\makeatletter
\renewcommand\paragraph[1]{\@startsection{paragraph}{4}{\parindent}%
{0.1ex \@plus 0.1ex \@minus 0.2ex}% small space above
{-1em}% keep it inline
{\normalfont\itshape\normalsize}*{#1.\quad}}
\makeatother

\hypersetup{
    colorlinks=true,
    linkcolor=blue,
    filecolor=blue,
    urlcolor=blue,
    citecolor=blue
}

\definecolor{lightblue}{RGB}{45,170,225}
\definecolor{bluecustom}{RGB}{15,112,183}
\definecolor{yellowcustom}{RGB}{242,231,109}
\definecolor{greencustom}{RGB}{18,161,154}

\usetikzlibrary{positioning}

\begin{document}
    \title[How Numerical Errors Scale with Hardware]
    {\Large How Numerical Errors Scale with Hardware}
    \author{Niklas Baumgarten}
    \date{\today}

    \maketitle

    \let\thefootnote\relax\footnotetext{\texttt{niklas.baumgarten@uni-heidelberg.de}}

    \maketitle

    \begin{abstract}
        In this short technical note, we outline how to bound numerical errors
        with respect to computational resources.
        We consider algorithms with known complexity estimates that
        can be executed in parallel and combine these estimates with parallelization laws.
        We suppose that the proposed approach generalizes
        beyond the presented examples
        and argue that it provides a straightforward and
        fair framework for comparing algorithms,
        hardware systems, and software implementations.
        This is achieved by imposing identical
        resource or energy budgets on each configuration,
        and comparing the resulting numerical error evaluations.
    \end{abstract}

    \section{Introduction}\label{sec:introduction2}

Energy consumption and the availability of
hardware resources are among the main limiting factors in
achieving better computational results.
This is particularly true for highly parallel applications,
where large computing systems work toward a common goal,
such as estimating quantities of interest for highly uncertain processes,
simulating fluid flow in three spatial dimensions,
or training a large language model.

To describe the performance of algorithms on parallel hardware,
Amdahl's law~\cite{amdahl1967validity} characterizes the maximum
speedup achievable for a fixed-size problem when only a portion
of the computation can be parallelized.
Gustafson's law~\cite{gustafson1988reevaluating} revisits
this analysis by assuming that the problem size
scales with the number of processing units,
deriving a corresponding scaled speedup relative to serial execution.

For many numerical algorithms in scientific computing and machine learning,
however, the problem size is related nonlinearly to
the numerical error one seeks to minimize.
Neural scaling laws~\cite{bahri2024explaining}, for example,
investigate how the number of model parameters and the amount
of training data influence the training loss.
When the amount of training data and the number of model
parameters are scaled in tandem with a compute-optimal strategy,
the training loss can be expressed directly as a function
of the total training compute~\cite{kaplan2020scaling}.

In this note, we want to outline how to connect parallelization laws
with error-vs-cost relations, formally described by complexity analysis.
This creates a direct connection between hardware
resources and achievable accuracies,
and thereof a direct relation between errors and invested energy.
For this we establish a relative error reduction,
rather than the speedup as in~\cite{amdahl1967validity, gustafson1988reevaluating},
achieved by executing an algorithm on a parallel system.
The here outlined proposition summarizes and generalizes observations made
in~\cite{baumgarten2025multilevel,baumgarten2025budgeted,baumgarten2024fully}.
%to algorithms with power-laws as complexity.

The approach is motivated by a knapsack problem:
minimize a numerical error,
subject to the constraints imposed by the computing system.
These constraints typically include the reserved computing time,
the available memory or data and the number of processing units such as CPUs, GPUs, or TPUs.

\clearpage

\begin{problem}[Minimal Error Knapsack]
    \label{problem:minimal-error-knapsack}
    Find an algorithm $\set{\pi_k}_{k}$ (a sequence of decisions)
    that minimizes some error estimate $\widehat{\mathrm{err}}_k$
    of the problem of interest within given computational constraints
    \begin{align*}
        \min_{\set{\pi_k}_{k}} \, \widehat{\mathrm{err}}_k
        \,\,\, \text{ s.t. } \,\,\,
        &\rC^{\mathrm{Op}}_k
        \leq P \cdot T, \,\,\,\,
        \rC_k^{\mathrm{Mem}} < M
        \,\,\, \text{ and } \,\,\,
        \rC_k^{\mathrm{Data}} < D \,.
    \end{align*}
    We denote by $\rC^{\mathrm{Op}}_k$ the operational cost,
    by $\rC^{\mathrm{Mem}}_k$
    the maximal memory footprint over the runtime,
    and by $\rC^{\mathrm{Data}}_k$ the amount of
    data processed up until step $k$.
    The constraints comprise a time budget $T > 0$,
    a fixed number of computing units $P \in \NN$,
    the total memory capacity of the system $M > 0$
    and the total amount of available data $D \geq 0$.
\end{problem}

Finding solutions to the resource
allocation~\Cref{problem:minimal-error-knapsack}
through an optimal sequence of decisions $\set{\pi_k}_k$
is an NP-hard optimization problem,
likely with non-linear objective functions and constraints.
Expressing this non-linear relationship between the error $\widehat{\mathrm{err}}_k$
(e.g., an estimator variance, a finite element error, or a training loss)
and the operational complexity $\rC^{\mathrm{Op}}_k$
(e.g., the number of necessary floating point operations, or the number of function evaluations)
is commonly done though complexity analysis identifying bounds of the form:
For any fixed accuracy $\epsilon \in (0, 1]$,
the expression $f \colon (0, 1] \rightarrow \RR_{\geq 0}$
bounds the operational complexity
\begin{equation}
    \label{eq:complexity-result}
    \rC_{\epsilon}^{\mathrm{Op}} \leq c \, f(\epsilon)
    \quad \text{ to achieve } \quad
    \widehat{\mathrm{err}} < \epsilon \, .
\end{equation}
The constant $c > 0$ is generally not known and
may dependent on the condition-number of the problem as
well as the actual implementation of the algorithm.
The function $f$ is assumed to be monotonically increasing as
$\epsilon \rightarrow 0$, i.e., achieving smaller errors is
associated with higher computational costs.
This is a natural result everywhere
in machine learning or scientific computing,
causing long computing times
and large hardware requirements in practice.
Complexity analysis allows us to quantify how good an algorithm
$\set{\pi_k}_k$ is in terms of the resources it needs to achieve a certain accuracy $\epsilon \in (0, 1]$,
enabling comparisons of different algorithms
with respect to their resource requirements.

The memory and data constraints in~\Cref{problem:minimal-error-knapsack}
typically determine the smallest achievable error for a given algorithm.
For example, in machine learning, limited training data introduces bias or
may lead to overfitting.
Similarly, insufficient memory can prevent the construction of a
larger finite element basis in numerical simulations.
Despite their importance, memory and data constraints
are here not the primary investigation goal,
we will come back to this only shortly in a final discussion.

    \section{Proposition}
\label{sec:proposition}

While very useful, complexity analysis alone does not
address key questions in practice:
Which error can be achieved within a given time frame?
How much can additional hardware resources improve the result?
Or, what is the smallest possible error that can be achieved on a given system?
This becomes even less clear if we consider the
influence of parallelization inefficiencies and
non-uniform memory architectures.

Addressing the above requires precise knowledge
of the constant $c$ in~\eqref{eq:complexity-result}.
In practice, however, this constant is usually unknown
due to incomplete information about the problem of interest
and the complete implementation stack.
By instead considering~\Cref{problem:minimal-error-knapsack},
we can circumvent knowing constant $c$ in~\eqref{eq:complexity-result},
at the cost of designing algorithms $\set{\pi_k}_k$
that implement greedy strategies or dynamic programming techniques.
For algorithms that (i) solve~\Cref{problem:minimal-error-knapsack},
(ii) are implemented in parallel,
and (iii) for which the dominant term in the
complexity estimate is multiplicative and
can be distributed across processing units,
we state the following.

\begin{proposition}
    \label{proposition:budgeted-error-bound}
    Let $\lambda_{s} \in (0, 1]$ and $\lambda_{p} \in [0, 1)$
    represent the non-parallelizable and parallelizable fractions
    of algorithm $\set{\pi_k}_k$ such that $\lambda_{s} + \lambda_{p} = 1$.
    Suppose this algorithm follows the operational complexity~\eqref{eq:complexity-result}
    and that $f \colon (0, 1] \rightarrow \RR_{\geq 0}$ is multiplicative (e.g.~a power-law), i.e.,
    \begin{equation}
        \label{eq:multiplicative-complexity}
        f(\epsilon_1 \cdot \epsilon_2) = f(\epsilon_1) \cdot f(\epsilon_2)
    \end{equation}
    and invertible, i.e., $f^{-1} \colon \RR_{\geq 0} \rightarrow (0, 1]$ exists.
    Then
    \begin{equation}
        \label{eq:budgeted-result}
        \widehat{\mathrm{err}}_k
        \,\,\, \lesssim \,\,\,
        {\lambda_{s} f^{-1}(T)}
        + \lambda_{p} f^{-1}(P \cdot T) \,,
    \end{equation}
    if every processing unit can handle a sufficiently large
    portion of the dominant term in the complexity estimate.
\end{proposition}

\begin{proof}
    Following~\cite{baumgarten2024fully, baumgarten2025multilevel},
    we present generalized arguments.
    Let $\rC_{\epsilon,\mathrm{ref}}^{\mathrm{Op}}$
    denote the serial execution time required to reach accuracy $\epsilon \in (0, 1]$.
    Similar to~\cite{amdahl1967validity, gustafson1988reevaluating},
    we decompose the execution time into a parallelizable $\lambda_{p} \in [0, 1)$
    and a non-parallelizable fraction $\lambda_s = (1 - \lambda_{p})$ by
    \begin{align*}
        \rC_{\epsilon,\mathrm{ref}}^{\mathrm{Op}}
        = (1 - \lambda_{p}) \rC_{\epsilon,\mathrm{ref}}^{\mathrm{Op}}
        + \lambda_{p} \rC_{\epsilon,\mathrm{ref}}^{\mathrm{Op}} \,.
    \end{align*}
    On $P \in \NN$ processing units, the execution time reduces to
    \begin{align*}
        \rC_{\epsilon,\mathrm{par}}^{\mathrm{Op}}
        = (1 - \lambda_{p}) \rC_{\epsilon,\mathrm{ref}}^{\mathrm{Op}}
        + \lambda_{p} f^{*}(P) \rC_{\epsilon,\mathrm{ref}}^{\mathrm{Op}} \,,
    \end{align*}
    where $f^{*} \colon \NN \to (0,1]$ is monotone decreasing,
    speeding up the computation to reach accuracy $\epsilon \in (0,1]$.
    The resulting reduction factor is
    \begin{align*}
        R
        \coloneqq \rC_{\epsilon,\mathrm{par}}^{\mathrm{Op}} / \rC_{\epsilon,\mathrm{ref}}^{\mathrm{Op}}
        = (1 - \lambda_{p}) + \lambda_{p} f^{*}(P) \leq 1 \,.
    \end{align*}
    Now, we fix a time budget $T > 0$ and search for the achievable
    accuracies $\epsilon_{\mathrm{ref}}$ on serial and $\epsilon_{\mathrm{par}}$
    on parallel hardware.
    Inverting~\eqref{eq:complexity-result} yields
    $\epsilon_{\mathrm{ref}} \lesssim f^{-1}(T)$ on serial hardware.
    For parallel hardware, using the reduction $R$ gives
    \begin{align*}
        \epsilon_{\mathrm{par}}
        \lesssim R \cdot f^{-1}(T)
        = (1 - \lambda_{p}) f^{-1}(T) + \lambda_{p} f^{*}(P) f^{-1}(T).
    \end{align*}
    The stated bound follows by~\eqref{eq:multiplicative-complexity}
    and if $f^*(P)=f^{-1}(P)$ on $\NN$,
    which holds when each processing unit handles a
    sufficiently large share of the total workload.
\end{proof}

\Cref{proposition:budgeted-error-bound} combines complexity estimates of the
form~\eqref{eq:complexity-result} with parallelization laws.
It assumes that $\epsilon_{\mathrm{par}}$ scales similarly to $R \cdot \epsilon_{\mathrm{ref}}$
and that processing units are interchangeable with computing
time in the parallelizable part of the complexity.
This behavior has been observed in practice
(e.g.~\cite{baumgarten2024fully, baumgarten2025multilevel})
and we conjecture that it also holds for other algorithms,
discussing applications in the next section.
Before doing so, we draw attention to~\Cref{fig:figure},
based on~\cite[Fig.~7]{baumgarten2025multilevel}
and~\cite[Fig.~12]{baumgarten2024fully}.

\begin{figure}
    \centering
    \begin{tikzpicture}[scale=0.55]
    \scriptsize{
        \draw[->] (0,0) -- (6,0) node[right] {};
        \draw[->] (0,0) -- (0,5) node[left] {$\widehat{\mathrm{err}}_k$};

        \draw[thick] (0,4.5) -- (5,3.0) node[right] {$P_0$};
        \draw[thick] (0,3.7) -- (5,2.2) node[right] {$2 \, P_0$};
        \draw[thick] (0,3.0) -- (5,1.5) node[right] {\dots};
        \draw[thick] (0,2.4) -- (5,0.9) node[right] {\dots};
        \draw[thick] (0,1.9) -- (5,0.4) node[right] {$16 \, P_0$};

        \newcommand{\xval}{2.8}
        \pgfmathsetmacro{\yA}{4.5 - 0.3*\xval}
        \pgfmathsetmacro{\yB}{3.7 - 0.3*\xval}
        \pgfmathsetmacro{\yC}{3.0 - 0.3*\xval}
        \pgfmathsetmacro{\yD}{2.4 - 0.3*\xval}
        \pgfmathsetmacro{\yE}{1.9 - 0.3*\xval}
        \draw[dashed] (\xval,\yA) -- (\xval,\yB);
        \draw[dashed] (\xval,\yB) -- (\xval,\yC);
        \draw[dashed] (\xval,\yC) -- (\xval,\yD);
        \draw[dashed] (\xval,\yD) -- (\xval,\yE);

        \fill (\xval,\yA) circle (2pt);
        \fill (\xval,\yB) circle (2pt);
        \fill (\xval,\yC) circle (2pt);
        \fill (\xval,\yD) circle (2pt);
        \fill (\xval,\yE) circle (2pt);

        \node[left] at (0.1+\xval,{(\yA+\yB)/2+0.1}) {\tiny{$\Delta_1$}};
        \node[left] at (0.1+\xval,{(\yD+\yE)/2+0.1}) {\tiny{$\Delta_2$}};

        \draw (0.0,-0.04) -- (0.0,0.08) node[below] {$0$};
        \draw (5.0,-0.04) -- (5.0,0.08) node[below] {$T$};
    }
    \end{tikzpicture}
    \hspace{0.5cm}
    \begin{tikzpicture}[scale=0.55]
        \scriptsize{
            \draw[->] (0,0) -- (6,0) node[right] {};
            \draw[->] (0,0) -- (0,5) node[left] {$\widehat{\mathrm{err}}_k$};

            \draw (0.0,-0.04) -- (0.0,0.08) node[below] {$P_0$};
            \draw (1.25,-0.04) -- (1.25,0.08) node[below] {};
            \draw (2.5,-0.04) -- (2.5,0.08) node[below] {};
            \draw (3.75,-0.04) -- (3.75,0.08) node[below] {};
            \draw (5.0,-0.04) -- (5.0,0.08) node[below] {$16\,P_0$};

            \newcommand{\xvall}{5.0}
            \pgfmathsetmacro{\yA}{4.5 - 0.3*\xvall}
            \pgfmathsetmacro{\yB}{3.7 - 0.3*\xvall}
            \pgfmathsetmacro{\yC}{3.0 - 0.3*\xvall}
            \pgfmathsetmacro{\yD}{2.4 - 0.3*\xvall}
            \pgfmathsetmacro{\yE}{1.9 - 0.3*\xvall}

            \fill (0.0,\yA) circle (2pt);
            \fill (1.25,\yB) circle (2pt);
            \fill (2.5,\yC) circle (2pt);
            \fill (3.75,\yD) circle (2pt);
            \fill (5.0,\yE) circle (2pt);

            \def\a{-1.571924}
            \def\b{4.606112}
            \draw[dashed,samples=100,domain=0:5]
            plot ({\x},{\a+\b*pow(2,-0.24*\x)});
        }
    \end{tikzpicture}
    \vspace{-0.1cm}
    \caption{
        \underline{Left:} Example of log-log error decay over time points up until $T$
        with a power law as complexity function in~\eqref{eq:complexity-result}
        and for different values of $P\in\set{P_0,2P_0,\dots,16P_0}$, $P_0 \in \NN$. \\[1mm]
        \underline{Right:} Achieved error estimates $\widehat{\mathrm{err}}_{k}$
        within a fixed time budget $T$ and line fit as dashed line to~\eqref{eq:error-with-hardware}.
    }
    \label{fig:figure}
\end{figure}

The left illustration in~\Cref{fig:figure} captures the
diminishing error reductions obtained with additional hardware.
This can be seen by comparing the distance $\Delta_1$ between the lines
for $P_0$ and $2P_0$ with the distance $\Delta_2$ between those for $8P_0$ and $16P_0$.
The latter is smaller, $\Delta_2 < \Delta_1$,
which can be explained by~\eqref{eq:budgeted-result}:
as $P$ increases, the first term, which is independent of $P$,
becomes increasingly important.
To illustrate this, the right panel of~\Cref{fig:figure}
takes the achieved error estimates $\widehat{\mathrm{err}}_k$ at
the fixed final time $T$ and plots them against
$P \in \set{P_0,2P_0,\dots,16P_0}$.
Equation~\eqref{eq:budgeted-result} then becomes
\begin{equation}
    \label{eq:error-with-hardware}
    \widehat{\mathrm{err}}_k \lesssim \epsilon_s + c_p \, f^{-1}(P) \,,
\end{equation}
where $\epsilon_s$ and $c_p$ can be determined by regression,
since $f^{-1}$ is either known theoretically or observed
empirically as a function of $T$.
This regression yields an error within the
measurement noise if $f^*(P)=f^{-1}(P)$ and
thus, observing the right plot of~\Cref{fig:figure}
confirms the proposition and that its assumptions are given.

In~\cite{baumgarten2024fully}, it is motivated to call the term $\epsilon_s$
parallelization bias, as $P \rightarrow \infty$, the bound~\eqref{eq:budgeted-result}
reduces to $\widehat{\mathrm{err}}_k \lesssim \lambda_s f^{-1}(T) \eqqcolon \epsilon_s$.
Thus, regardless of how much hardware is available,
the first summand can only be reduced by
increasing the available computing time.
In practice, the NP-hardness of~\Cref{problem:minimal-error-knapsack}
can be addressed with greedy algorithms or dynamic programming,
both of which require the algorithm
to exchange information at some point,
implying $\lambda_s \neq 0$ and thereof $\epsilon_s \neq 0$.

As a consequence of bound~\eqref{eq:budgeted-result},
we can also impose energy budgets, i.e.,
given a fixed total amount of Watt-hours,
we identify the best algorithm, hardware system,
or software implementation by comparing the achieved errors.
Here, the first term on the right-hand side of~\eqref{eq:budgeted-result}
is associated with the energy consumption of the predominantly idle system,
whereas the second term directly translates the number of
processing units into power requirements.
Based on electricity prices and energy mix,
even monetary or carbon dioxide
budgets may be imposed.

    \section{Examples}\label{sec:examples}

The observation made in~\Cref{proposition:budgeted-error-bound}
can be applied to a wide range of algorithms.
We shortly discuss four examples.

\begin{example}[Monte Carlo Estimation]
    Let $X$ be a random variable with unknown finite variance.
    A Monte Carlo method to estimate the expected value $\EE[X]$
    is given by
    \begin{equation}
        \label{eq:mc-estimator}
        E_{M,s}^{\mathrm{MC}}[X]
        = \frac{1}{M} \sum_{m=1}^M X^{(m)},
        \,\,
        \text{ with } X^{(m)} \text{ i.i.d. .}
    \end{equation}
    The root mean squared error
    follows by the central limit theorem $\cO(1/\sqrt{M})$,
    i.e., $f(\epsilon) = \epsilon^{-2}$.
    Estimator~\eqref{eq:mc-estimator} can be implemented
    greedily and embarrassingly parallel to
    solve~\Cref{problem:minimal-error-knapsack}
    and thus satisfies all assumptions
    in~\Cref{proposition:budgeted-error-bound}.
    Numerical verification of this can
    be found in~\cite{baumgarten2024fully},
    applying multilevel Monte Carlo methods
    to acoustic wave equations.
\end{example}

\begin{example}[Gradient Descent]
    Let $x^* = \mathrm{argmin}_{x \in \RR^N} F(x)$.
    If $F$ is $\mu$-strongly convex, $L$-smooth
    and for sufficient step size $\tau > 0$,
    the iteration scheme
    \begin{equation}
        \label{eq:gradient-descent}
        x^{(k+1)} = x^{(k)} - \tau \nabla F(x^{(k)})
    \end{equation}
    converges linearly, and achieving some $\epsilon \in (0, 1]$, results in
    $\cO(\log (1 / \epsilon))$ number of steps~\cite[Sec.~9.3]{boyd2004convex}.
    Since this does not satisfy the assumption and~\eqref{eq:gradient-descent}
    is a sequential iteration scheme, \Cref{proposition:budgeted-error-bound} is not directly applicable.
    However, if computational resources also have to be invested
    into other tasks, e.g~function evaluation or batch estimation,
    and aligning the iteration index of~\eqref{eq:gradient-descent}
    with algorithm $\set{\pi_k}_k$,~\Cref{proposition:budgeted-error-bound}
    again becomes applicable.
    For example, estimate~\eqref{eq:budgeted-result}
    is observed in~\cite{baumgarten2025multilevel}
    for variants of batched stochastic gradient descent.
\end{example}

\begin{example}[Finite Element Methods]
    FEMs approximate functions on bounded spatial domains $\cD \subset \RR^d$
    by decomposing them into finitely many cells of diameter $h > 0$,
    resulting in large algebraic systems.
    The discretization error typically decays as $\cO(h^{\alpha})$,
    where the rate $\alpha > 0$ depends on the regularity of the solution
    and the chosen discretization scheme,
    whereas the computational cost, assuming optimal solvers,
    scales as $\cO(h^{-d})$.
    Achieving a FEM error below $\epsilon \in (0,1]$
    therefore requires computational work following
    $f(\epsilon) = \epsilon^{-d/\alpha}$ (see, e.g.,~\cite{cliffe2011multilevel}).
    However, finding an algorithm that implements a FEM and
    solves~\Cref{problem:minimal-error-knapsack} is challenging.
    In practice, the mesh size $h$ is typically refined geometrically,
    so the constraints must be adapted accordingly.
    For example, halving $h$ requires eight folding
    $P$ for $d=3$ and a fixed $T > 0$.
    Consequently, solving~\Cref{problem:minimal-error-knapsack}
    amounts to a classical weak scaling experiment.
    Despite the fact that FEM and
    solving~\Cref{problem:minimal-error-knapsack}
    can not be combined in a flexible way,
    embedding FEM within statistical estimators or
    gradient descent methods make it possible
    again to apply~\Cref{proposition:budgeted-error-bound}
    as done in~\cite{baumgarten2025multilevel, baumgarten2024fully}.
\end{example}

\begin{example}[Training of Neural Networks]
    The neural scaling law~\cite[Eq.~(1.3)]{kaplan2020scaling},
    setting the training loss $L$ and a compute budget $C_{\min}$
    in relation through the rate $\alpha_C^{\min} > 0$,
    \begin{equation}
        \label{eq:neural-scaling-law}
        L(C_{\min}) \propto (C_{\min})^{-\alpha_C^{\min}}
    \end{equation}
    has been found through balancing dataset size and trainable parameters,
    optimally allocating the compute budget, which is related
    to solving~\Cref{problem:minimal-error-knapsack}.
    Equation~\eqref{eq:neural-scaling-law} presents a similar
    as~\eqref{eq:budgeted-result}, however,
    we argue that separating $C_{\min}$ into
    $P \cdot T$ presents deeper insights how additional hardware
    actually improves the computational result.
    Hence, we argue that~\Cref{proposition:budgeted-error-bound}
    is applicable to neural scaling laws for compute.
\end{example}

    \section{Conclusion and Discussion}\label{sec:discussion}
\Cref{problem:minimal-error-knapsack},
which is commonly encountered in machine learning and scientific computing,
can be viewed similarly to approaches in the
Theory of Constraints~\cite{goldratt2016goal, csimcsit2014theory}. % ronen1992managing
It identifies a performance measure (a numerical error),
as well as system constraints
(e.g., available processing units, allocated computing time,
memory capacity, data availability, or even energy budgets).
The validity of~\Cref{proposition:budgeted-error-bound}
requires examination beyond the presented examples,
but it provides a straightforward argument for whether and how
lifting constraints improves computational results.

This note focuses on exploiting the computing-time constraint
by combining complexity analysis with parallelization laws,
leading to~\Cref{proposition:budgeted-error-bound}.
Depending on the application, memory capacity may instead be
the limiting factor; see, e.g.,~\cite{baumgarten2025budgeted}.
Similarly, data availability may become the dominant constraint,
as outlined, for example, by scaling laws~\cite{kaplan2020scaling}
with respect to dataset size.
With this note, we advocate for approaching computational problems
from the perspective of physical constraints, detailing in particular
how increasing the number of computing units affects numerical errors.

    \bibliographystyle{ieeetr}

    \bibliography{bibliography}

\end{document}